\documentclass[11pt]{article}
\usepackage[margin=1in]{geometry}
\usepackage{amsmath,amssymb,amsthm}
\usepackage{cite}
\usepackage{graphicx}
\usepackage{hyperref,aliascnt}
\usepackage{microtype}

\newtheorem{theorem}{Theorem}

\newaliascnt{lemma}{theorem}
\newtheorem{lemma}[lemma]{Lemma}
\aliascntresetthe{lemma}

\newaliascnt{observation}{theorem}
\newtheorem{observation}[observation]{Observation}
\aliascntresetthe{observation}

\theoremstyle{definition}
\newaliascnt{definition}{theorem}
\newtheorem{definition}[definition]{Definition}
\aliascntresetthe{definition}

\begin{document}

\title{NC Algorithms for Computing a Perfect Matching and\\ a Maximum Flow
in One-Crossing-Minor-Free Graphs\thanks{A preliminary version of this paper appears in SPAA 2019.}}

\author{David Eppstein and Vijay V. Vazirani\\
Computer Science Department, University of California, Irvine}

\date{ }

\maketitle

\begin{abstract}
In 1988, Vazirani gave an NC algorithm for computing the number of perfect matchings in $K_{3,3}$-minor-free graphs by building on Kasteleyn's scheme for planar graphs, and stated that this ``opens up the possibility of obtaining an NC algorithm for finding a perfect matching in $K_{3,3}$-free graphs.'' In this paper, we finally settle this 30-year-old open problem. Building on recent NC algorithms for planar and bounded-genus perfect matching by Anari and Vazirani and later by Sankowski, we  obtain NC algorithms for perfect matching in any minor-closed graph family that forbids a one-crossing graph. This family includes several well-studied graph families including the $K_{3,3}$-minor-free graphs and $K_5$-minor-free graphs.  Graphs in these families not only have
unbounded genus, but can have genus as high as $O(n)$. Our method applies as well to several other problems related to perfect matching.
In particular, we obtain NC algorithms for the following problems in any family of graphs (or networks) with a one-crossing forbidden minor:
\begin{itemize}\itemsep 0pt
\item Determining whether a given graph has a perfect matching and if so, finding one.
\item Finding a minimum weight perfect matching in the graph, assuming that the edge weights are polynomially bounded.
\item Finding a maximum $st$-flow in the network, with arbitrary capacities.
\end{itemize}
The main new idea enabling our results is the definition and use of \emph{matching-mimicking networks}, small replacement networks that behave the same, with respect to matching problems involving a fixed set of terminals, as the larger network they replace. 
\end{abstract}

\section{Introduction}

Obtaining an NC algorithm for matching has been an outstanding
open question in theoretical computer science for over three decades, ever since the discovery of RNC~matching algorithms \cite{KarUpfWig-Comb-86, MulVazVaz-Comb-87}.
In a recent breakthrough result, Anari and Vazirani gave an NC algorithm for finding a perfect matching in planar graphs \cite{AnaVaz-FOCS-18}. 
Subsequently, Sankowski provided an alternative algorithm for the same problem based on different techniques~\cite{San-ICALP-18}.
By using a reduction from flow problems on other surfaces to planar flow~\cite{BorEppNay-SCG-16},
Anari and Vazirani also extended their result to graphs of bounded genus. 
Their paper restated the open problem of obtaining an NC algorithm for finding a
perfect matching in $K_{3,3}$-minor-free graphs, in particular because such graphs
can have genus as high as $O(n)$. This problem was previously stated by 
Vazirani in a 1989 paper in which he gave an NC algorithm for computing the number of 
perfect matchings in such graphs~\cite{Vaz-IC-89} and
stated that this ``opens up the possibility of obtaining an NC algorithm for finding a perfect 
matching in $K_{3,3}$-free graphs.''\footnote{Vazirani stated this problem in terms of graphs with no subgraph homeomorphic to $K_{3,3}$, rather than with no $K_{3,3}$ minor. However, for 3-regular graphs such as $K_{3,3}$, subgraphs homeomorphic to $H$ are equivalent to $H$-minors. We use the formalism of minors because it fits better with our generalization to other forbidden minors. We vary from Vazirani in using the terminology ``$K_{3,3}$-minor-free'' rather than ``$K_{3,3}$-free'' to avoid confusion with a third, unrelated meaning, the graphs that do not have $K_{3,3}$ as induced subgraphs.}  In this paper, we finally settle this 30-year-old open problem. 

\begin{figure}[t]
\centering\includegraphics[scale=0.45]{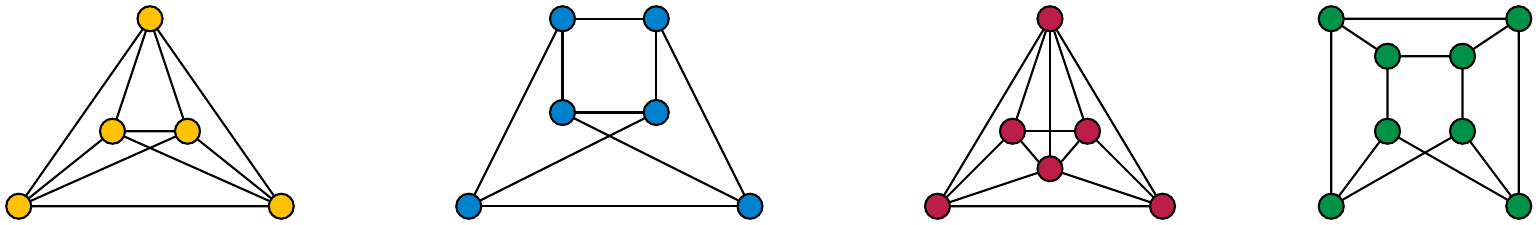}
\caption{Forbidden minors for previously-studied one-crossing-minor-free graph classes, and their one-crossing drawings: $K_5$ (yellow), $K_{3,3}$ (blue), $K_6-2e$ (red), and the Wagner graph (green)}
\label{fig:k5-k33-w8}
\end{figure}

The $K_{3,3}$-minor-free graphs are particularly attractive as a target for this problem because they form a natural extreme case for certain approaches. In particular, they are known to have Pfaffian orientations, by which their matchings can be counted using matrix determinants~\cite{Kas-GTTP-67,Lit-CM-74}, while for $K_{3,3}$ itself and for any minor-free family that does not forbid it, this tool is unavailable.
However, our result breaks through this barrier: we give an NC algorithm for finding a perfect matching in graphs belonging to any one-crossing-minor-free class of graphs.
That is, if $H$ is any graph that can be drawn in the plane with only one crossing pair of edges,
then we can find perfect matchings in the $H$-minor-free graphs in NC.
Because $K_{3,3}$ can be drawn with one crossing, our result includes in particular the 
$K_{3,3}$-minor-free graphs. Other previously-studied graph classes to which our result applies include the $K_5$-minor-free graphs~\cite{Wag-MA-37}, the $(K_6-2e)$-minor-free graphs~\cite{Mah-JGT-08}, and the Wagner-minor-free graphs~\cite{MahRob-JCTB-16}; see \autoref{fig:k5-k33-w8} for the forbidden minors of these classes.

The one-crossing-minor-free graphs are not themselves a class of individual graphs, but of graph classes. They are contained in the family of minor-closed graph classes, and they contain both the crossing-free (planar) and the bounded-treewidth minor-closed graph classes; see \autoref{fig:venn}. One-crossing-minor-free graphs are noteworthy, among minor-closed graph classes, for obeying a simpler structure theorem. The graphs in any minor-closed graph class can be decomposed into subgraphs of bounded genus using several operations called ``clique-sums'', ``apices'', and ``vortices''. If a graph class has a forbidden minor with one crossing, this decomposition can be simplified to use only planar subgraphs (instead of more general bounded-genus subgraphs), and only clique-sum operations (instead of the other two operations)~\cite{RobSey-GST-91}. In this way, algorithms on one-crossing-minor-free graphs are motivated not just by the specific classes of graphs to which they apply, but as a step towards handling the operations needed for the full minor-closed structural decompositions. For this reason, several past works have studied algorithms for one-crossing-minor-free classes of graphs~\cite{ChaEpp-JGAA-13,DHN-JCSS-04,DHT-Algo-05,Kam-TCS-12,StrThiWag-TCS_16}.

\begin{figure}[t]
\centering\includegraphics[scale=0.45]{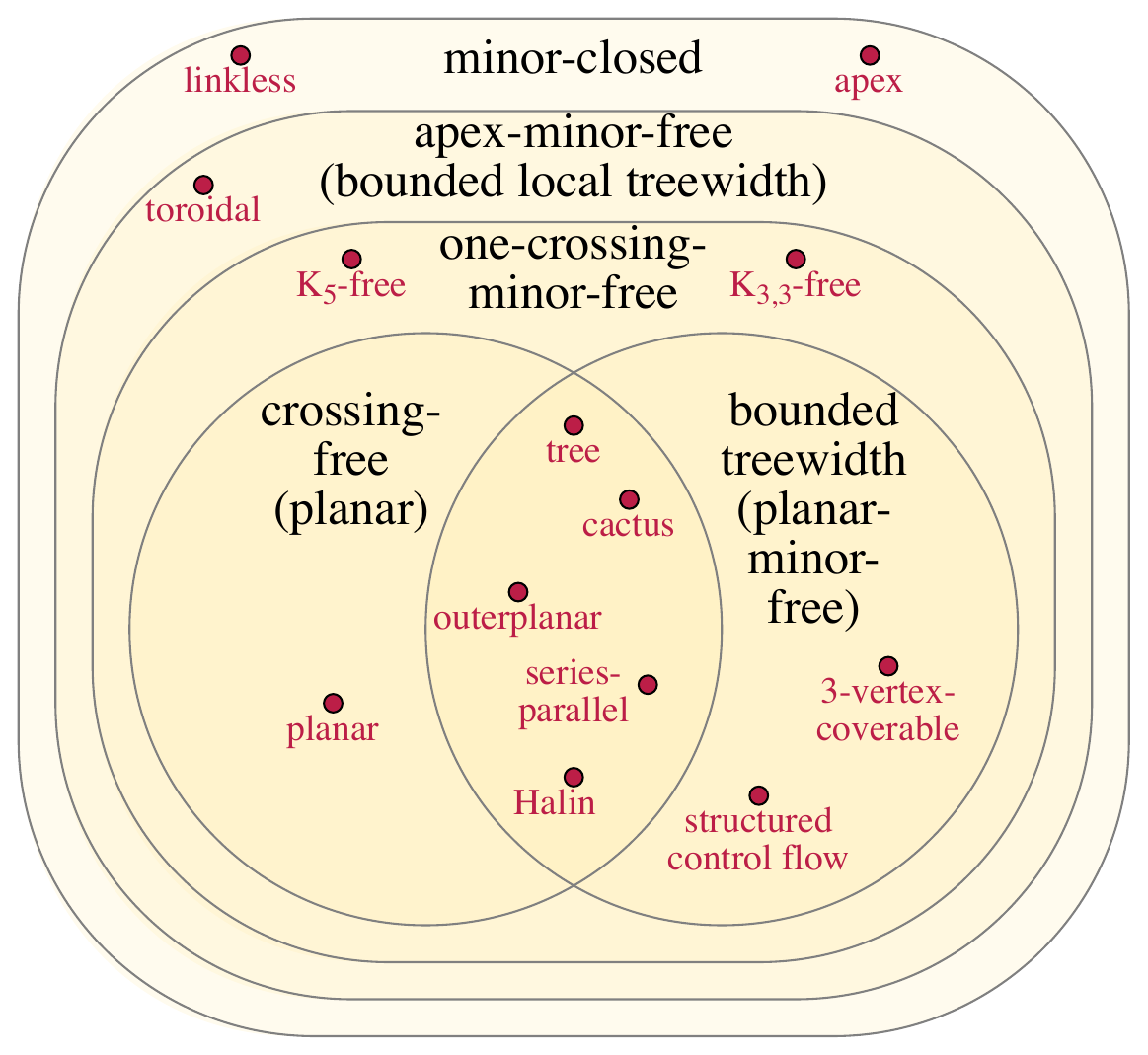}
\caption{A Venn diagram of minor-closed classes of graphs (red points), and properties of these classes (yellow outlined regions), including the one-crossing-minor-free classes and the apex-minor-free classes. Our results apply to all graph classes within the ``one-crossing-minor-free'' region.}
\label{fig:venn}
\end{figure}

We build on our methods for finding perfect matchings to give NC algorithms that, on the same one-crossing-minor-free families of graphs, compute a perfect matching of minimum weight 
when the weights are polynomially-bounded integers.
In another direction, we obtain an NC algorithm for finding a maximum $st$-flow in any 
flow network whose underlying undirected graph belongs to a one-crossing-minor-free family.
This generalizes Johnson's 1987 result~\cite{Joh-JACM-87}, showing that 
maximum $st$-flow in a planar network is in NC; we note that this result was used crucially in the planar graph perfect matching NC algorithm of \cite{AnaVaz-FOCS-18}. 

\subsection{Technical ideas}

Our main new technical idea is that of a \emph{matching-mimicking network}.
Given a graph $G$ and a set $T$ of terminal vertices, a matching-mimicking network is a
graph $G'$, containing $T$, that has the same pattern of matchings:
every matching of $G$ that covers $G\setminus T$ corresponds to a matching of $G'$ that
covers $G'\setminus T$ and vice versa.

We show that matching-mimicking networks exist for graphs with any bounded number of terminals.
The size of these networks is bounded by a function of the number of terminals. For at most three terminals, we describe these networks explicitly. In this case, the mimicking networks are planar and remain planar when glued into the triangular face of a larger planar network, and can be given edge weights so that their minimum-weight perfect matchings (for each subset of terminals) mimic the weights of the minimum-weight matchings in the given graph. Both of these properties, their planarity and their weight-mimicking ability, are needed by our algorithm. 

In the past, mimicking networks for network flow were defined and used by numerous 
researchers for obtaining flow algorithms~\cite{HagKatNis-JCSS-98,ChaSubWag-Algo-00,KraRik-SODA-12,ChaEpp-JGAA-13,KhaRag-IPL-14}.
These mimicking networks were first defined to prove that maximum flow can be found
in NC in graphs of bounded treewidth~\cite{HagKatNis-JCSS-98}
and later used also in efficient sequential algorithms for flow in
one-crossing-minor-free graphs~\cite{ChaEpp-JGAA-13}. Their theoretical properties have also become an object of study in their own right~\cite{ChaSubWag-Algo-00,KraRik-SODA-12,KhaRag-IPL-14}.
It seems likely that, similarly, our matching-mimicking networks will lead to algorithmic applications beyond our NC matching algorithm, and additional theory beyond our existence proof.

As with a previous sequential flow algorithm of Chambers and Eppstein~\cite{ChaEpp-JGAA-13}, we exploit the structural decomposition of graphs with a one-crossing forbidden minor~\cite{RobSey-GST-91}, by repeatedly using mimicking networks to simplify this structure. However the order in which we perform these replacements must be more carefully chosen so that our algorithms run in NC.
Each step of the replacement process involves the computation of matchings either in a 
bounded-treewidth graph or in a planar graph. The planar matchings can be found by the new results of Anari and Vazirani or of Sankowski, and the bounded-treewidth matchings can be found in NC by using log-space versions of Courcelle's theorem~\cite{ElbJakTan-FOCS-10}.\footnote{More specifically, Theorem 1.3 of Elberfeld, Jaoby, and Tantau~\cite{ElbJakTan-FOCS-10} provides a logarithmic space algorithm, for any monadic second order graph property with a free set variable, that counts the sets of each cardinality for which a given bounded-treewidth graph models the formula. Matching may be expressed in this way: a set $S$ of edges forms a matching if the graph models a formula stating that no two edges of $S$ share a vertex. The claim follows from the inclusion of LOGSPACE in NC.}

Our maximum $st$-flow result uses a similar algorithmic outline, with flow-mimicking networks 
in place of matching-mimicking networks. Our method differs from the sequential algorithm of 
Chambers and Eppstein~\cite{ChaEpp-JGAA-13}, which used flow-mimicking networks on at most three terminals to replace leaf nodes of the structural decomposition tree. Our matching algorithms, also, use matching-mimicking networks on at most three terminals, replacing subtrees of more than one node in a single step. However, our parallel flow algorithm uses flow-mimicking networks for a second purpose, namely to replace a component of our matching algorithm that involves semi\-ring matrix multiplication. This part of our flow algorithm requires flow-mimicking networks on up to six terminals.

\subsection{History and related results}
 
In a seminal paper, Lovasz~\cite{Lov-FCT-79} proposed a way of computing a perfect matching using
methods quite different from the combinatorial (augmenting-path-finding) methods that were 
the mainstay at the time. His method used linear algebra and randomization; the connection to linear algebra being established via the Tutte matrix of the given graph. Although not mentioned explicitly in this paper, it was clear that his methods gave an RNC algorithm for the decision problem of determining if a given graph has a perfect matching.

When combinatorial methods were found to be lacking for obtaining a fast parallel 
matching algorithm, researchers turned to Lovasz's proposed method.
The first RNC~algorithm for finding a perfect matching was
obtained by Karp, Upfal, and Wigderson~\cite{KarUpfWig-Comb-86}. 
This was followed by a somewhat simpler 
algorithm due to Mulmuley, Vazirani, and Vazirani~\cite{MulVazVaz-Comb-87}; their RNC algorithm also extends to finding a minimum weight perfect matching in case the edge-weights are polynomially bounded.

Matching has played a central role in the development of the theory of algorithms, in that its
study, from various computational viewpoints, has led to quintessential paradigms and powerful 
tools for the entire theory. These include the notion of polynomial time solvability \cite{Edmonds65}, the counting class \#P \cite{Valiant79} and a polynomial time equivalence between random generation and approximate counting for self-reducible problems \cite{JVV}, which lies at the core of the Markov chain Monte Carlo method. The perspective of parallel algorithms has also led to such a gain, namely the Isolation Lemma \cite{MulVazVaz-Comb-87}, which has found several applications in complexity theory and algorithms.

However, this development still did not clarify whether randomization was essential for fast 
parallel matching. Considering the fundamental insights gained by an algorithmic study of
matching and the possibility of further insights by setting this question, this has remained a significant open question ever since the 1980s.
The first substantial progress on this question was made by Miller and Naor in 1989~\cite{MillerN}.
They obtained an NC algorithm for finding a maximum flow from a set of sources to a set of
sinks in a planar network; as a corollary, they obtained an NC~algorithm for finding a 
perfect matching in bipartite planar graphs.
As is well known, Kasteleyn's algorithm for counting the number of perfect matchings
in a planar graph~\cite{Kas-GTTP-67} can be easily made into an NC algorithm for counting 
perfect matchings by using Csanky's NC
algorithm for the determinant of a matrix~\cite{Csa-SICOMP-76}. In 2000, Mahajan and Varadarajan gave an elegant way of using this NC algorithm for counting 
perfect matchings to find one, 
hence giving a different NC algorithm for bipartite planar graphs~\cite{MahajanV}.
 (The NC counting algorithm was also used
critically in Anari and Vazirani's NC algorithm for non-bipartite planar graphs~\cite{AnaVaz-FOCS-18}.)

After a decade and half of lull, there has been a resurgence of activity on this
problem over the last five years. In particular, several researchers have obtained quasi-NC 
algorithms for matching and its generalizations. Such an algorithm
runs in polylogarithmic time; however, it requires $O(n^{\log^{O(1)} n})$ processors. All the 
algorithms in this line of research work by a partial derandomization of the Isolation Lemma. 
This line of work was started by Fenner, Gurjar, and Thierauf, who gave a quasi-NC algorithm 
for perfect matching in bipartite graphs \cite{FenGurThi-CACM-19}. Later,
Svensson and Tarnawski extended the result to general graphs \cite{SveTar-FOCS-17}. 
The generalization of bipartite matching to the linear matroid intersection problem was
given by Gurjar and Thierauf~\cite{GurjarT} and to  
finding a vertex of a polytope with totally unimodular constraints by Gurjar, Thierauf, and 
Vishnoi~\cite{GurjarTV}. 
Researchers have also developed pseudo-deterministic RNC matching algorithms, which output the same (unique) solution for almost all choices of random bits~\cite{GolGro-ICALP-17,AnaVaz-19}. 

Recently, Anari and Vazirani \cite{AnaVaz-19}  have given what appears to be the culmination of this line of work: An NC algorithm for finding a minimum weight perfect matching in a general graph with polynomially bounded edge weights, provided it is given an oracle for the decision problem. The latter problem is: given a general graph with polynomially bounded edge weights and a target weight $W$, determine if there is a perfect matching of weight at most $W$ in the graph. As mentioned in \cite{AnaVaz-19}, many of the ideas discovered in the last five years figured in a key way in their work. In a similar vein, it is not unlikely that results such as ours, which extend the frontier of NC matching algorithms, are likely to play a critical role towards the eventual resolution of the full problem.

The first NC algorithm for finding a maximum $st$-flow in a planar network was obtained
by Johnson~\cite{Joh-JACM-87}. As stated above, this was followed by an NC algorithm for
finding a maximum flow from a set of sources to a set of sinks in a planar network 
by Miller and Naor~\cite{MillerN}. An NC algorithm for maximum flow 
in graphs of bounded treewidth was given by Hagerup et al.~\cite{HagKatNis-JCSS-98}.
An NC algorithm for counting perfect matchings in the same classes of graphs that we study, the one-crossing-minor-free classes of graphs, was given recently by Straub et al.~\cite{StrThiWag-TCS_16}; by applying a reduction of Kulkarni et al.~\cite{KulMahVar-CJTCS-08} they were also able to construct perfect matchings in NC for bipartite graphs in these families, but not in the general case.
our algorithm has a similar overall structure to theirs, based on finding the structural decomposition of one-crossing-minor-free graph classes and then replacing small components of the decomposition by gadgets.

\section{Matching-Mimicking Networks}

\begin{definition}
A \emph{matching}, in an undirected graph, is a subset of edges no two of which share a vertex.
The matching \emph{covers} a subset of vertices, the ones that are vertices of the selected edges.
If $G$ is a graph with a specified subset $T$ of vertices,
we define the \emph{matching pattern} of $G$ to be a family of subsets of $T$,
the subsets $X\subset T$ such that some matching of $G$ covers $(G\setminus T)\cup X$ (and covers no other vertices).
If $G$ and $G'$ are two graphs, both containing a shared subset $T$ of vertices,
we say that $G$ and $G'$ are \emph{matching-equivalent} on $T$
if they have the same matching patterns.
A \emph{matching-mimicking network} for $G$ and $T$ is any other graph $G'$ containing $T$ that 
is matching-equivalent to $G$ on $T$.
\end{definition}

\begin{lemma}
\label{lem:mimic-size-func}
There is a function $f$ such that any graph $G$ and subset of vertices $T$ has a 
matching-mimicking network of at most $f(|T|)$ vertices.
\end{lemma}

\begin{proof}
This follows immediately from the fact that $T$ has $2^{|T|}$ subsets,
and therefore that $G$ has at most $2^{2^{|T|}}$ matching patterns.
For each matching pattern that can be achieved, we may take the mimicking network to be the smallest graph that achieves that matching pattern.
The worst-case size of the resulting mimicking network is the largest size of a finite set of graphs (one for each achievable matching pattern) and is therefore bounded as a function of $|T|$.
\end{proof}

This result is not very explicit, as we do not have an algorithm for determining whether a matching pattern can be achieved nor for finding the smallest graph that achieves it. Therefore, carrying out this method of constructing matching-mimicking networks requires a case analysis to determine which patterns are achievable and which graphs achieve them. It would be of interest to find an explicit algorithm for constructing a matching-mimicking network of bounded size. In contrast, for flow-mimicking networks, the following algorithm can be used: for each nontrivial partition of the terminals, find a minimum cut in the network that separates the two sides of the partition, identify two vertices of the network as equivalent when they are on the same side of every cut, and collapse each equivalence class to a single vertex.

\begin{figure}[t]
\centering
\includegraphics[width=0.9\textwidth]{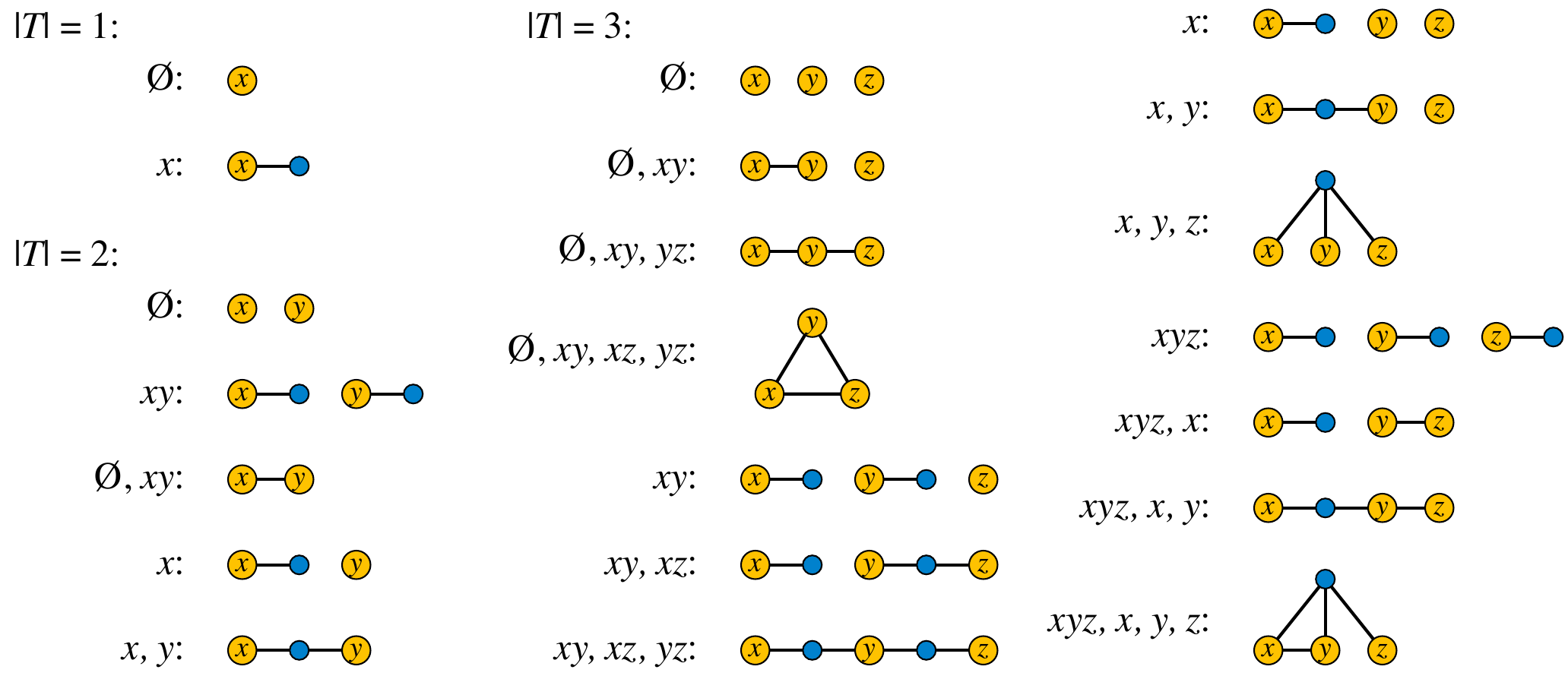}
\caption{Matching-mimicking networks for $|T|\le 3$. In each network, the set $T$ of terminals consists of the labeled yellow vertices; the remaining non-terminal vertices of the network are the smaller blue vertices. The labels denote the subsets of terminals that can be covered by a matching.}
\label{fig:matching-mimic}
\end{figure}

We have performed by hand the case analysis needed to construct matching-mimicking networks  for $|T|\le 3$, the largest number of terminals needed for our algorithms. The results are depicted in \autoref{fig:matching-mimic}. As the figure shows, for $|T|=3$ the number of matching patterns is 14, much smaller than the $2^{2^3}=256$ bound on the number of patterns obtained by plugging $|T|=3$ into the proof of \autoref{lem:mimic-size-func}. We achieve this reduction in the number of cases by combining the following three observations:
\begin{itemize}
\item We omit graphs whose matching pattern is empty. If any such graph is detected during
our algorithm for matching, we may abort the algorithm, as the whole graph has no matching.
\item The sizes of the subsets of~$T$ in any single matching pattern must all have the same parity as each other.
\item We may consider matching patterns to be equivalent whenever one matching pattern can be obtained from another by permuting the vertices of~$T$. We only need to find matching-mimicking networks for each equivalence class of matching patterns.
\end{itemize}

The following property indicates that, when glued into a planar graph, all of the matching-mimicking networks of the figure preserve its planarity. We need this property in our algorithm,
so that we can continue to compute matchings in the result of such gluings.

\begin{lemma}
\label{lem:planarity-preserving}
Let $G$ be a planar graph, let $f$ be a triangular face of a planar drawing of $G$, and let $T$ be a subset of the vertices of $f$. Then the union of $G$ and any of the matching-mimicking networks of \autoref{fig:matching-mimic}, with terminal set $T$, is another planar graph that can be drawn
in the plane with the matching-mimicking network inside~$f$.
\end{lemma}

\begin{proof}
This follows from the layouts given for these networks in the figure, which are all drawn outerplanar (planar and with all vertices belonging to the unbounded face of the drawing).
Because they are outerplanar, their unbounded face can be surrounded by triangle $f$ and then, around $f$, the rest of $G$, without creating any new crossings.
\end{proof}

\section{Structural Decomposition}
\label{sec:struc-decomp}

A $k$-clique-sum of two graphs is defined as a graph that can be obtained from the disjoint union of the two given graphs by identifying a clique of $\le k$ vertices in one of the graphs with a clique of the same size in the other graph, and then optionally deleting some of the edges of the merged clique. One-crossing-minor-free graphs have a structural decomposition that can be described in terms of clique-sums: If $H$ is a graph that can be drawn in the plane with at most one edge crossing, then the $H$-minor-free graphs can be decomposed by $3$-clique-sums into pieces that are either planar or of bounded treewidth~\cite{RobSey-GST-91}. This decomposition generalizes the result that graphs with a planar forbidden minor have bounded treewidth, and
is a simplified form of the structural decomposition of arbitrary minor-closed graph families by Robertson and Seymour, which also includes pieces of bounded genus, apexes (vertices that can be adjacent to any subset of the other vertices in a single piece), and vortexes (subgraphs of bounded pathwidth attached to a face of a bounded-genus piece).
The graphs with one-crossing drawings include $K_{3,3}$ and $K_5$ (\autoref{fig:k5-k33-w8}), whose corresponding minor-free graph classes have even simpler forms of this decomposition: the $K_{3,3}$-minor-free graphs are $2$-clique-sums of planar graphs and $K_5$, and the $K_5$-minor-free graphs are $3$-clique-sums of planar graphs and the eight-vertex Wagner graph (shown in \autoref{fig:k5-k33-w8}, right).

\begin{figure}[t]
\centering\includegraphics[width=0.6\textwidth]{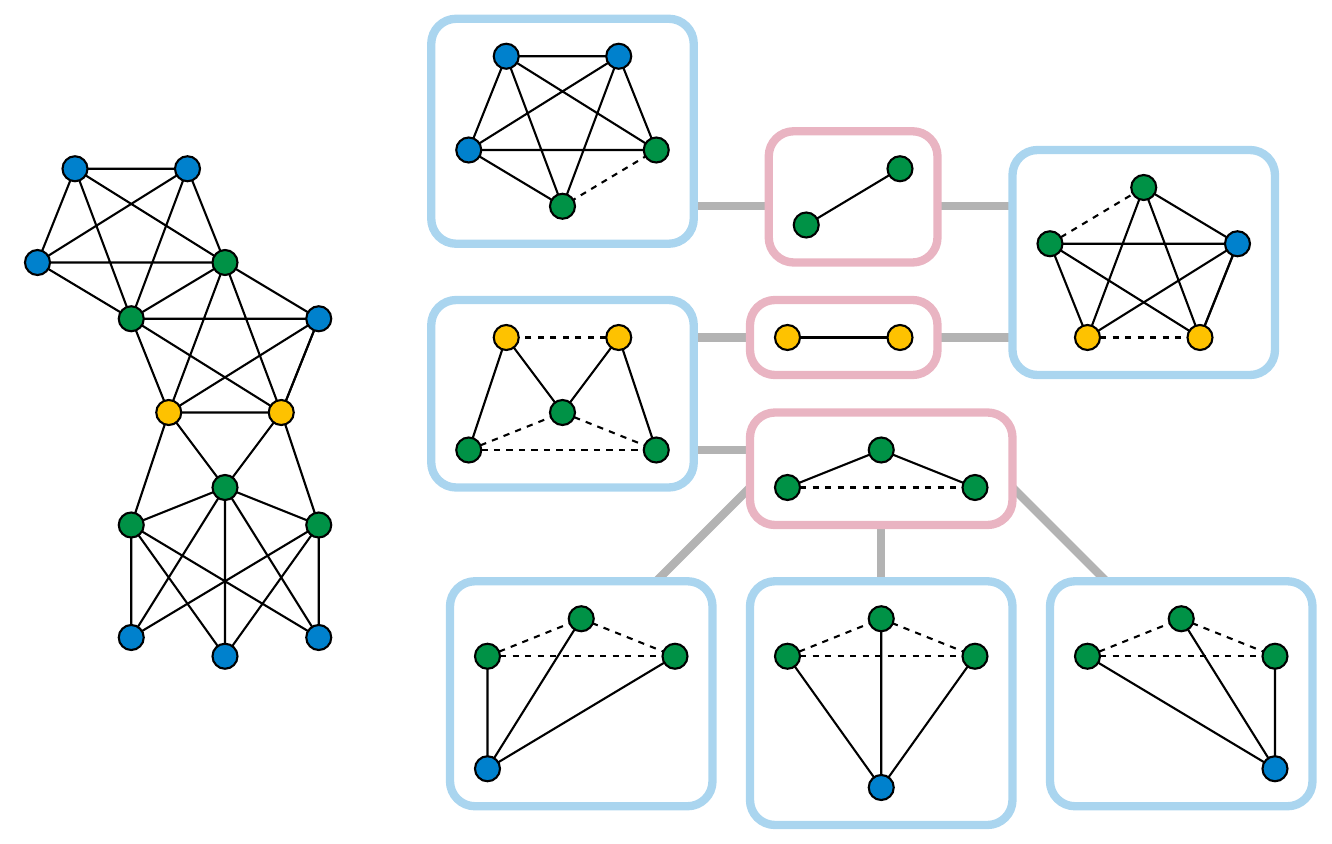}
\caption{Decomposition of a graph into a $3$-clique-sum of simpler pieces (light blue outline) and cliques (light red outline). The dashed edges indicate edges of glued cliques that are to be removed after gluing, either because they are redundant (the same edge appears elsewhere) or because they
are not part of the original graph. Note that this decomposition is not \emph{maximal}, in the sense of \autoref{lem:laminar}, as the piece with two yellow vertices and three green vertices can be partitioned by another separator.}
\label{fig:decomposition}
\end{figure}

We can describe any clique-sum decomposition, such as the decomposition from this structure theorem, as a two-colored tree (\autoref{fig:decomposition}), in which the nodes of one color represent pieces (planar graphs or bounded-treewidth graphs), and the nodes of the other color represent cliques on which two are more pieces are glued. The edges of this tree describe the incidence relation between edges and cliques. Additionally, each clique node of the decomposition tree is labeled with information describing which of its edges are kept as part of the overall graph.
When a planar piece is glued to other graphs along a non-facial triangle, we may split the planar piece into two smaller pieces on that triangle. Therefore, it is safe to assume that, for each planar piece of the decomposition, the 3-vertex cliques incident to it are all faces of a planar embedding of the piece.

For our algorithms, it is necessary not merely to know that this $3$-clique-sum decomposition exists, but also to find it, in NC.
Efficient decomposition algorithms are known for $K_{3,3}$-minor-free graphs and for $K_5$-minor-free graphs~\cite{Asa-TCS-85,ReeLi-LATIN-08}, but they are sequential, and we are not aware of such algorithms for the general case. In previous work on sequential flow algorithms we avoided this issue by assuming that the decomposition was given as part of the input~\cite{ChaEpp-JGAA-13}. Fortunately, in our new results, we do not need the decomposition to be efficient; it merely needs to be in NC.

\begin{definition}
A subset of three vertices is a \emph{separator} if its removal would increase the number of connected components of the remaining graph. We define a family of separators to be \emph{laminar} when
no two vertices from any one separator in the family are separated by any other separator in the family.
\end{definition}

We are interested in finding a decomposition of a given graph, assumed to be from a one-crossing-minor-free family, by minimal separators of at most three vertices.
However, this decomposition is not unique; for instance $K_{3,3}$ has two incompatible minimal separators, the two sides of its bipartition.  There exist graphs (such as the wheel graph) for which there are quadratically many separating triples, most pairs of which are non-laminar, but maximal laminar sets of separating triples in planar graphs may be found sequentially in linear time~\cite{EppRee-SODA-19}. Our graphs are non-planar, and again we are only interested in membership in NC.

\begin{lemma}
\label{lem:laminar}
We can find a maximal laminar family of separators of size $\le 3$, in any graph, in NC.
\end{lemma}

\begin{proof}
In parallel, list all subsets of at most three vertices, and check whether each one is a minimal separator. Create an \emph{incomparability graph}, in which the vertices represent separators,
and the edges represent non-laminar pairs of separators. Find a maximal independent set in the incomparability graph. All of the steps of testing whether subsets are minimal separators or whether two separators are non-laminar involve connectivity computations in undirected graphs, which can be done in NC or even LOGSPACE~\cite{Rei-STOC-05}. The incomparability graph has polynomial size, and a maximal independent set in any graph can be constructed in NC~\cite{Lub-SICOMP-86}.
\end{proof}

\begin{definition}
A $Y$--$\Delta$ transformation, of a graph $G$ at a degree-three vertex $v$, consists of removing $v$ from $G$ and replacing it by a triangle connecting its three neighbors.
\end{definition}

The following two results on $Y$--$\Delta$ transformations are folklore:

\begin{lemma}
\label{lem:planar-yd}
If $G$ is planar then so is the result of any set of $Y$--$\Delta$ transformations of $G$.
\end{lemma}

\begin{proof}
Let such a transformation replace $v$ by a triangle $abc$.
From a planar drawing of $G$, we may obtain a planar drawing of the transformed graph, by routing edges $ab$, $bc$, and $ac$ along curves in the plane near the paths $avb$, $bvc$, and $avc$ respectively.
\end{proof}

A $Y$--$\Delta$ transformation can increase the treewidth of a graph; for instance, the claw $K_{1,3}$ has treewidth one but its $Y$--$\Delta$ transformation, the triangle $K_3$, has treewidth two. Nevertheless, the increase is not great:

\begin{lemma}
\label{lem:treewidth-yd}
If $G$ has treewidth $w$ then the result of any set of $Y$--$\Delta$ transformations of an independent set of vertices of $G$ has treewidth $O(w)$.
\end{lemma}

\begin{proof}
Consider any tree-decomposition of $G$ of width $w$; this is a tree whose vertices, called \emph{bags}, are associated with sets of at most $w+1$ vertices of $G$, such that each vertex of $G$ is contained in the bags for a connected subtree of the given tree, and such that each edge of $G$ has both endpoints contained in at least one bag.
For each vertex $v$ replaced in a $Y$--$\Delta$ transformation, choose two of its three neighbors,
and in each bag containing $v$ replace $v$ by those two chosen neighbors. The result is a tree-decomposition of the transformed graph of width $O(w)$.
\end{proof}

\begin{lemma}
\label{lem:decompose}
We can find a structural decomposition of graphs in any one-crossing-minor-free family,
into pieces that are either planar or have bounded treewidth (but do not necessarily themselves belong to the family), in NC.
\end{lemma}

\begin{proof}
Let $G$ be a graph in the given family.
We first find a maximal laminar family of separators of $G$ of size at most one using \autoref{lem:laminar}.  By performing additional connectivity computations, we find the pieces that they separate the graph into and the tree $T_1$ of clique-sums by which these pieces can be glued to form $G$. Each piece in this decomposition is an induced subgraph of $G$, so it still belongs to the given family, and is in addition 2-vertex-connected.

Within each piece, we next find a maximal laminar family of separators of size at most two, the pieces they separate the graph into, and the tree of clique-sums for these pieces. By 2-vertex-connectivity, each 2-clique-sum in each piece connects it to a connected subgraph of $G$, and contracting that subgraph into a single edge produces the extra edge added to the piece to represent the 2-clique-sum. Therefore, each piece in this decomposition is a minor of $G$, so it still belongs to the given family. In addition, each piece is 3-vertex-connected, for otherwise we would have found additional two-vertex separators in our laminar family. By replacing each node of $T_1$ by the tree of clique-sums obtained in this way, we obtain a tree $T_2$ of 1- and 2-clique-sums by which these 3-vertex-connected minors of $G$ can be glued together to obtain $G$ itself.

Finally, within each 3-vertex-connected piece $P$, we find a maximal laminar family of separators of size at most three, the pieces they separate $P$ into, and a tree of clique-sums for these pieces.
For each piece $Q$ of the decomposition of piece $P$, let $Q'$ be the graph obtained as a minor of $P$ by contracting each connected component of $Q\setminus P$ into a single vertex. By 3-vertex-connectivity, this contraction process produces a single vertex for each 3-clique-sum involving $Q$, adjacent to the three vertices of the 3-clique-sum.
Then $Q$ itself can be found by performing for each of these vertices a $Y$--$\Delta$ transformation that replaces this degree-3 vertex by a triangle, connecting the three vertices of the 3-clique-sum and making them into a clique.
Because $Q'$ is an indecomposable graph in the given minor-closed family, it must either be planar or have bounded treewidth. And because $Y$--$\Delta$ transformations preserve both planarity (\autoref{lem:planar-yd}) and bounded treewidth (\autoref{lem:treewidth-yd}), it follows that every piece $Q$ of the resulting decomposition is itself either planar or of bounded treewidth.
By replacing each node of $T_2$ corresponding to a piece $P$ by the tree of clique-sums obtained in this way, we obtain a tree of $k$-clique-sums for $k\le 3$, decomposing the original graph $G$ into pieces that are either planar or of bounded treewidth.
\end{proof}

Additionally we can check which pieces are planar, and find a planar embedding for the planar pieces, in NC~\cite{JaSim-SICOMP-82}.

\begin{figure}[t]
\centering\includegraphics[width=0.75\textwidth]{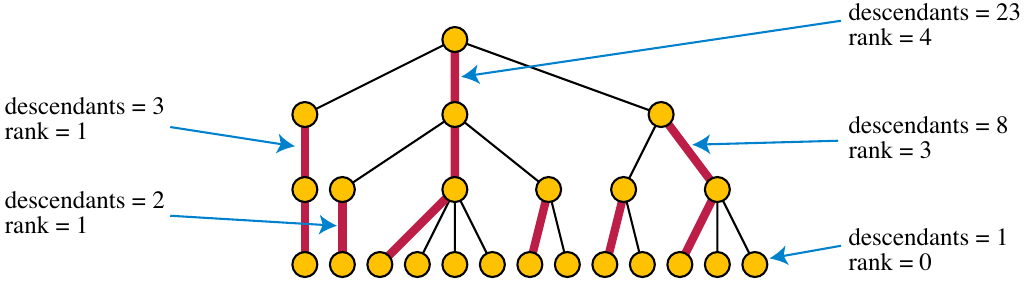}
\caption{Heavy path decomposition of a rooted tree, showing the number $k$ of descendants of the top node of each path (including itself) and the rank $\lfloor\log_2 k\rfloor$ of the path.}
\label{fig:heavypath}
\end{figure}

It will be convenient to define one more tool, a structural decomposition of our structural decomposition. It is the \emph{heavy path decomposition} of a tree (the tree describing the structural decomposition).

\begin{definition}
If any tree $T$ is given an arbitrary root, it may be decomposed into paths by choosing at each non-leaf node of the tree a single child, the one with the most descendants (counting each node as one of its own descendants, and choosing arbitrarily in case of ties). The chosen parent-child edges link together to form a cover of $T$ by vertex-disjoint paths, including some length-zero paths for unchosen leaf vertices. These paths are the \emph{heavy paths} of the decomposition. Each heavy path (other than the one containing the root of the tree) has a \emph{parent path}, the path containing the parent of the topmost vertex in the path. We may define the \emph{rank} of a heavy path whose topmost vertex has $k$ descendants to be $\lfloor\log_2 k\rfloor$.
\end{definition}

An example is shown in \autoref{fig:heavypath}.

Then in a heavy path decomposition of a tree with $n$ nodes, all ranks are integers in the range from $0$ to $\lfloor\log_2 n\rfloor$. If $P$ is any heavy path that does not contain the root of the tree, the rank of $P$ is strictly less
than the rank of the parent of $P$, because $P$ must have at most half as many descendants as its parent (if it had more, it would have been picked as the heavy child from the parent of the top node of $P$). The heavy path decomposition was introduced for its applications in sequential data structures~\cite{SleTar-JCSS-83,HarTar-SJC-84}, and has become a standard tool for graph drawing and geometric graph algorithms~\cite{DunEppGoo-DCG-13,EppGoo-TC-11}, but it has also been recently applied in parallel algorithms~\cite{FisKopKur-CPM-16}. By using the Euler tour technique for trees~\cite{TarVis-SJC-85}, we may easily count the descendants of each node in a tree, obtaining the following result.

\begin{lemma}
\label{lem:heavy}
We may find the heavy path decomposition of any tree, together with the ranks of each of its paths, in NC.
\end{lemma}

\section{Perfect Matching Algorithm}

To find a perfect matching in a given graph $G$, from a one-crossing-minor-free family $\mathcal{F}$, we perform the following steps.

\begin{enumerate}
\item We apply \autoref{lem:decompose} to find a decomposition of $G$ into a 3-clique-sum of pieces that are labeled as either planar or of bounded treewidth.
\item We root the decomposition tree arbitrarily, and use \autoref{lem:heavy} to find a heavy path decomposition of the resulting rooted tree structure.
\item For each rank $r$ from $0$ to $\lfloor\log_2 n\rfloor$ of a path in the heavy path decomposition (sequentially), we perform the following steps.
\begin{enumerate}
\item In parallel, for each heavy path $P$ of rank $r$, we replace $P$ in the structural decomposition of $G$ by a single matching mimicking network. The terminals of this mimicking network are the (at most three) vertices of $G$ by which $P$ attaches to its parent in the structure tree,
and the graph it mimics is the one formed by the clique-sum of all pieces of the decomposition of $G$ that either belong to $P$ or descend from $P$. (We will describe how to construct this mimicking network below.)
\item We define a \emph{shallow clique} of the decomposition to be a node of the rooted decomposition tree, representing a clique at which two or more pieces are attached to each other, such that all child pieces are mimicking networks (rather than larger pieces or subtrees), and such that the parent piece is planar.
That is, because of priori replacements, the remaining parts of the decomposition tree below the shallow clique consist only of leaves of the tree, each of which has already been replaced by a mimicking network.
Necessarily, such a clique has size at most three.
At each shallow clique, in parallel, we replace the clique-sum of the attached mimicking networks (a graph of bounded treewidth) with a single mimicking network for the clique-sum, and then replace the parent piece with its clique-sum with this mimicking network,
removing the shallow clique and its descendants from the decomposition tree.
By \autoref{lem:planarity-preserving} this operation preserves the planarity of the parent piece.
Because we remove the clique from the decomposition tree, it also preserves the property that in planar pieces of the decomposition, all 3-clique-sums occur on face triangles.
\end{enumerate}
\item We reverse the sequence of replacements by matching-mim\-ic\-k\-ing networks,
maintaining throughout the reversed sequence a perfect matching for the current graph. When multiple pieces were replaced in parallel, we perform their reversed replacement in parallel in the same way. To reverse the replacement of a single graph by a mimicking network, given a matching in the mimicking network, we construct a corresponding matching in the graph that was replaced, covering the same terminal vertices and all nonterminal vertices.
\end{enumerate}

It remains to explain how to perform step 3(a), in which we construct the mimicking network for a heavy path, and also how we maintain enough information about how we constructed it to reverse the replacement of the path in step 4.

\begin{definition}
For each node interior to a heavy path $P$, define the two \emph{sides} of the node to be the two sets of at most three vertices by which its subgraph is connected to its neighbors in $P$. Similarly, for the topmost node of $P$ (the one closest to the root of the decomposition tree) we let one of the sides be the set of at most three vertices connecting it to its parent, and for the bottommost (leaf) node of $P$ we define one of its sides to be the empty set.
\end{definition}

the matchings that are possible for the subgraph associated with each interior node of a heavy path can
be summarized by a Boolean matrix, the \emph{transfer matrix} of the node (\autoref{fig:transfer}), defined below.
These matrices are not themselves mimicking networks (they are matrices, not graphs), but we will use them in our computation of mimicking networks for heavy paths.

\begin{figure}[t]
\centering
\includegraphics[scale=0.4]{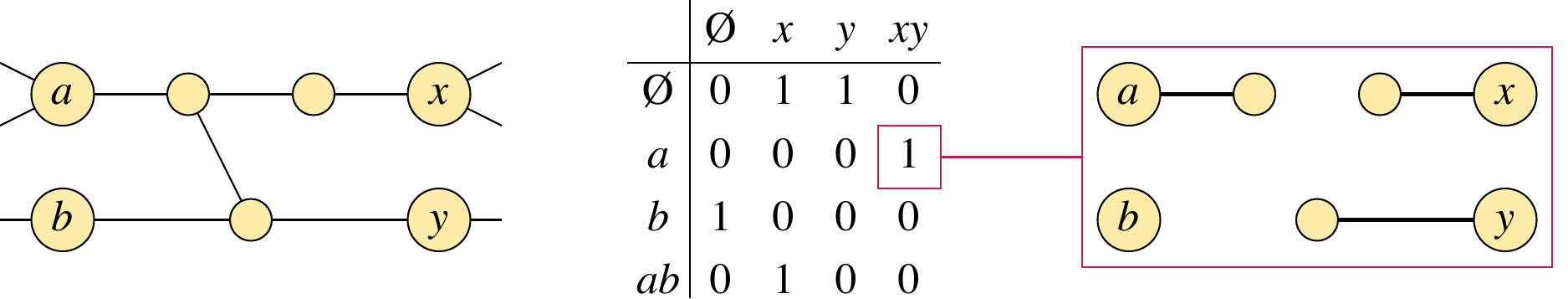}
\caption{A subgraph with sides $\{a,b\}$ and $\{x,y\}$, its transfer matrix, and a matching corresponding to one of the nonzero entries in the matrix.}
\label{fig:transfer}
\end{figure}

\begin{definition}
The transfer matrix is a Boolean matrix whose rows are indexed by subsets of the vertices on one side of the node, the side closest to the leaf of the path, and whose columns are indexed by subsets of the vertices on the other side. The entry of this matrix in row $i$ and column $j$ is true if there exists a matching of the subgraph associated with the node that covers all interior vertices of the subgraph,
and covers the subsets of the two sides indexed by $i$ and $j$.
Otherwise, the entry is false.
\end{definition}

Taking into account the fact that, within a single side of each node, the subsets of vertices that can be matched are constrained to all have the same parity as each other,\footnote{The parity of the set of matched vertices on one side of a node must equal the parity of the set of vertices in the input graph minus the node, on the same side. For otherwise, not every vertex in that component could be matched.} these matrices have dimension at most $4\times 4$, but are in some cases smaller. For instance, the transfer matrix of the leaf node of the path has only one row, corresponding to the empty set, as its set of vertices on the leaf side is the empty set. 

\begin{lemma}
In the algorithm described above,
we can compute a mimicking network for each given heavy path, in NC. Moreover, given a matching in this mimicking network, we can compute a corresponding matching in the clique-sum
of the pieces in the path and its descendants, again in NC.
\end{lemma}

\begin{proof}
At the time our algorithm processes a heavy path $P$, all heavy paths of lower rank (and in particular all paths descending from $P$) will have already been processed.
Therefore, the nodes of $P$ will be of three types, each associated with a subgraph:
\begin{itemize}
\item Planar pieces of the decomposition, possibly with glued-in mimicking networks from lower-rank paths that preserve the planarity of the piece. We define the subgraph associated with the node to be this planar piece.
\item Bounded-treewidth pieces of the decomposition, possibly attached by clique-sums to mimicking networks from lower-rank paths. We define the subgraph associated with the node to be the clique-sum of it and its attached mimicking networks. Because it is a clique-sum of bounded-treewidth graphs, this associated subgraph has bounded treewidth.
\item Cliques of the decomposition, again possibly attached by clique-sums to mimicking networks from lower-rank paths. As with the bounded-treewidth pieces, we define the subgraph associated with the node to be the (bounded-treewidth) clique-sum of this clique with its attached mimicking networks.
\end{itemize}

To construct the mimicking network of the path, we first construct a transfer matrix for each node in the path.
Each Boolean value in each of these transfer matrices can be found in parallel by testing for the existence of a perfect matching in the induced subgraph of the vertices that should be covered.

The transfer matrix for the clique-sum of any contiguous subsequence of nodes in the path is just the product of its matrices, over the Boolean $(\vee,\wedge)$ semi\-ring, in left-to-right order from the leaf end to the root end. As with any product of matrices, we can compute the product matrix in NC,
for instance by associating the nodes of the subsequence with the leaves of a balanced binary tree and, at each interior node of the binary tree, multiplying the matrices from the two child nodes.
Because the leftmost (leaf) matrix is a row vector, the product of all the matrices will also be a row vector, indicating which subsets of the three vertices on the top side of the top node of $P$ can be covered by a matching that also covers all vertices belonging to the subgraphs associated with $P$.
The information in this row vector is exactly what we need to compute a mimicking network for $P$.

To make this process reversible, whenever we compute the product of two transfer matrices we also store, for each true Boolean value in the product matrix, a pair of true Boolean values in the two multiplicands that cause that product value to be true (choosing arbitrarily when multiple pairs would cause it to be true). Then, when we have chosen a matching in the mimicking network for $P$ and wish to replace it by a matching in the subgraphs associated with $P$, we trace back through this stored information to find a sequence of true Boolean values in the transfer matrices of each node of $P$ that together correspond to a matching of the correct type. Then, in each subgraph associated with a node of $P$, we perform a matching algorithm to find a single matching corresponding to the position of this value in its transfer matrix.
\end{proof}

\begin{theorem}
We can find a perfect matching in any graph from a one-crossing-minor-minimal graph family, in NC.
\end{theorem}

\begin{proof}
All mimicking networks used to replace other pieces in the decomposition tree have bounded size, and therefore bounded treewidth.
When we merge a shallow clique into its parent in the decomposition tree, we are gluing a single mimicking network for at most three terminals into a vertex, edge, or triangle of the parent graph. We only perform this merge step once per cut-vertex, separating edge, or separating triangle of the parent piece.
The parent piece must have been planar before the gluing step, and by \autoref{lem:planarity-preserving} it remains planar. Correspondingly, because this gluing step performs a $3$-clique-sum of a graph of bounded treewidth, if the parent piece was of bounded treewidth before the gluing step, it remains of bounded treewidth. Thus, all the replacements performed by the algorithm preserve the structure of the decomposition, allowing the algorithm to continue correctly in later steps.

There are logarithmically many iterations of the outer loop, and each iteration performs only steps that can be performed in NC. Therefore, the overall algorithm is also in NC.
\end{proof}

\section{Minimum Weight Perfect Matching Algorithm}

Our algorithm for finding minimum-weight matchings in one-crossing-minor-free graphs is similar in outline to the algorithm for unweighted matchings. We find a decomposition tree and its heavy path decomposition, and then for each path in rank order replace it by a mimicking network.
However, to apply this method to minimum-weight perfect matching, we need three additional ingredients:
\begin{itemize}
\item When we compute transfer matrices and their products, we replace the Boolean values in these matrices by numerical values, the minimum weight of each matching, and we replace the $(\vee,\wedge)$ semi\-ring used for existence of a matching with the $(\min,+)$ semi\-ring to compute the minimum weight of a matching. As before, whenever we multiply two of these matrices we store for each entry of the product the pair of entries of the multiplicands that gave rise to its value.
\item We need to be able to construct minimum-weight perfect matchings in the pieces of the structural decomposition, namely planar graphs and graphs of bounded treewidth. For planar graphs a minimum-weight perfect matching algorithm in NC (restricted to polynomially-bounded integer weights) was given by Anari and Vazirani~\cite{AnaVaz-FOCS-18} and by Sankowski~\cite{San-ICALP-18}. For bounded treewidth graphs it appears that the log-space version of Courcelle's theorem~\cite{ElbJakTan-FOCS-10} does not support optimization of structures expressible in monadic second-order logic, so it does not directly provide an algorithm for minimum-weight perfect matching in bounded-treewidth graphs. Nevertheless it is straightforward to obtain an NC algorithm for this problem directly, for instance by combining the known log-space tree-decomposition algorithm (which can be interpreted as decomposing any graph of bounded treewidth into a clique-sum of pieces of bounded size) with our method of heavy path decomposition and multiplication of transfer matrices along each heavy path.
\item We need planarity-preserving weighted matching-mimicking networks for sets of at most three terminals. We detail this ingredient below.
\end{itemize}

In order for a matching-mimicking network to preserve the choice of which perfect matching has minimum weight, we will require it to have the following property: let $M$ and $M'$ be two matchings in the original network that we wish to mimic, each covering all nonterminals and a different subset of terminals, and each of minimum weight for the subset of terminals that it covers.
Then the difference in weights between $M$ and $M'$ should be the same as the difference in weights between the corresponding two minimum-weight matchings in the mimicking network.
We do not require $M$ or $M'$ to have the same weight in the mimicking network as in the original network, but only that the two matchings differ by the same amount. This is because some of the matchings in our mimicking networks have an empty set of matched edges, and we cannot control the total weight of the empty set.

If a mimicking network had two different matchings covering the same subsets of terminals, we would have to worry about which of these two is the minimum-weight matching. Fortunately, this is not an issue:

\begin{observation}
For all of the three-terminal mimicking networks of \autoref{fig:matching-mimic}, and
all matchable subsets of terminals in each network, there is exactly one matching that covers that subset of terminals and all nonterminals.
\end{observation}

\begin{proof}
This uniqueness property is true when the mimicking network is a tree or a forest: Every leaf nonterminal must be matched and, regardless of whether we specify to match or not match a leaf terminal, there is only one way to do it. After this choice, the remaining subnetwork is still a tree or a forest, and the result follows by induction.

There are two remaining networks with cycles in  \autoref{fig:matching-mimic}, the one for $\{\emptyset,xy,xz,yz\}$ and the one for $\{xyz,x,y,z\}$. For the first of these two networks (a triangle $xyz$), the unique matching for $\emptyset$ is the empty matching, and the unique matching that covers any pair of terminals is the matching that uses the edge between those two terminals. For the second of these two networks, $z$ is a leaf, and any matching that covers $z$ covers its adjacent nonterminal, leaving a tree or forest network to which the reasoning above for trees and forests applies. The remaining cases for this network match one of $x$ or $y$, and can only do so by matching them to the unique nonterminal vertex.
\end{proof}

\begin{lemma}
\label{lem:min-weight-mimic}
For each of the three-terminal mimicking networks of \autoref{fig:matching-mimic},
it is possible to set weights on the edges of the network to preserve any given assignment of differences to the weights of its matchings.
\end{lemma}

\begin{proof}
When any of the matchings in one of these networks has an edge $e$ that is not used in any other of the matchings, we can set the weights of its other edges recursively, and then choose a weight for $e$ that causes its matching to have the correct difference with the other matchings. Using this strategy we can handle all of the mimicking networks in  \autoref{fig:matching-mimic}
that have at most two matchings (because surely each of the two has a uniquely used edge)
or that have no non-terminals (because in these networks, every non-empty matching is disjoint).
The remaining cases are:
\begin{itemize}
\item The network for the matching pattern $\{xy,xz,yz\}$. In this network, the two edges incident to $y$ in the figure are uniquely used in their two matchings.
\item The network for the matching pattern $\{xyz,x,y\}$. In this network,  let the fourth vertex of the network be $w$. Then the matching for $xyz$ uses edges $xw$ and $yz$, where $yz$ is uniquely used. The matching for $y$ uses edge $wy$, which is uniquely used. The matching for $x$ uses edge $xw$, which is not uniquely used. Since two of the three matchings have uniquely used edges, we can use the recursive strategy to assign weights to all the edges.
\item The network for the matching pattern $\{xyz,x,y,z\}$. In this network,  the triangle of edges connecting $x$, $y$, and the nonterminal vertex are uniquely used in their matchings.\qedhere
\end{itemize}
\end{proof}

Therefore, in all cases we can use the same three-terminal planarity-preserving matching-mimicking networks as in the unweighted case. These are all the ingredients that we need to prove the following result:

\begin{theorem}
We can find a minimum-weight perfect matching in any graph from a one-crossing-minor-minimal graph family, with polynomially-bounded weights, in NC.
\end{theorem}

\section{Maximum Flow Algorithm}

A similar algorithmic outline can be used to find a maximum $st$-flow in directed flow networks whose underlying undirected graphs belong to a one-crossing-minor-free graph family.
The differences are the following:
\begin{itemize}
\item We use the flow mimicking networks of Hagerup et al~\cite{HagKatNis-JCSS-98} in place of our new matching-mimicking networks. A flow-mimicking network, for a directed graph with specified edge capacities and specified terminal vertices, is another directed graph with specified edge capacities and the same set of terminal vertices,
such that there exists a flow with given supply and demand amounts at the terminals, obeying all the capacity constraints, in the original graph if and only if there exists a flow with the same supply and demand amounts in the mimicking network. One way of constructing mimicking networks networks is to find a system of minimum cuts separating each partition of the terminals, and to collapse together subsets of vertices that are not separated by any of these cuts, but as Chambers and Eppstein~\cite{ChaEpp-JGAA-13} showed, the networks for up to three terminals can be constructed directly and have the planarity-preserving properties required by our algorithm.
\item We use the same structural decomposition of the given flow network as we used in our matching algorithm, but we modify its heavy path decomposition. We define the root of the structural decomposition tree to be the entire path between the piece containing $s$ and the piece containing $t$ (where $s$ and $t$ are the two flow terminals). Based on this choice of root, we perform a heavy path decomposition of the remaining parts of the tree in the same way as before. Finally, we include the root path as a path in the decomposition, and we assign it a rank greater than that of any other path in the decomposition. In this way, the algorithm will simplify all other paths in the graph before reaching the root path.
\item When processing any single path in the heavy path decomposition, we have no way of using matrices to represent the flows through a subgraph associated with a node of the path.
Instead, we represent these flows with flow mimicking networks for at most six terminals, treating the vertices on  both sides of each node as terminals.
Instead of using matrix multiplication in a semi\-ring to combine pairs of matrices, we combine pairs of mimicking networks by constructing a single mimicking network for their union.
\item In the root path of the decomposition, the two end nodes of the path are the ones containing the terminals $s$ and $t$. Rather than defining one side of these nodes to be the empty set, we use the sets $\{s\}$ and $\{t\}$ respectively.  Alternatively, when the root path consists of a single node whose associated subgraph contains both $s$ and $t$, we use these two sets as its two sides.
In this way, the mimicking network constructed for this root path will consist of a single capacitated edge from $s$ to $t$, in which finding a maximum flow is trivial.
\end{itemize}

Hence we have the following result:

\begin{theorem}
We can find a maximum flow from $s$ to $t$ in any flow network in a one-crossing-minor-minimal graph family, with source $s$ and sink $t$, in NC.
\end{theorem}

\section{Discussion and Open Problems}

We conclude with some open problems that result from our work:

\begin{itemize}
\item Can we prove an explicit upper bound (for instance a closed form formula or primitive recursive function) on the size of matching-mimicking
networks, as a function of the number of terminals? Do weighted
matching-mimicking networks exist for arbitrary numbers of terminals?
\item Because Anari and Vazirani~\cite{AnaVaz-FOCS-18} show how to find perfect matchings in bounded-genus graphs,
our method immediately extends to facial $3$-clique-sums of bounded-genus graphs and bounded-treewidth graphs. However, it is not clear what happens when bounded-genus pieces are glued by clique-sums on triangles that are not faces. Do the facial $3$-clique-sums of bounded-genus graphs and bounded-treewidth graphs include any other natural graph classes?
\item Is it possible to use the structure theorem for more general minor-closed graph families,
allowing non-facial clique-sums, apexes, and vortexes, to find perfect matchings in such families in NC? A first step towards this generalization would be (following \autoref{fig:venn}) to consider the apex-minor-free graph families (families of graphs whose forbidden minors include at least one graph formed from a planar graph by adding one vertex). These families include the one-crossing-minor-free graph families, and like the one-crossing-minor-free graph families they have a useful structure that can be exploited in some algorithms, a functional relationship between their diameter and treewidth~\cite{Epp-Algo-00}.
\item Chambers and Eppstein~\cite{ChaEpp-JGAA-13} used mimicking networks to find near-linear-time sequential maximum flow algorithms for one-crossing-minor-free graphs. Our matching-mimicking networks could also be applied in finding sequential perfect matching algorithms for the same class of graphs. However, in order to obtain a speedup with this method, we need a subroutine for fast perfect matching in planar graphs. It is known that perfect matching in bipartite planar graphs can be solved in near-linear-time, by a reduction to flow with multiple sources and multiple sinks~\cite{BorKleMoz-SJC-17}. Are there similarly fast algorithms for non-bipartite planar perfect matching?
\item To what extent can our methods be extended from perfect matching to maximum matching? Datta et al~\cite{DatKulKum-ISAAC-18} have reduced the search problem to the decision problem for NC planar maximum matching, but it is still not known whether this problem is in NC.
\end{itemize}

\subsection*{Acknowledgements}

The research of David Eppstein was supported in part by NSF grants  CCF-1618301 and CCF-1616248.
The research of Vijay Vazirani  was supported in part by NSF grant CCF-1815901.
We thank Nima Anari for helpful discussions.

\bibliographystyle{plainurl}
\bibliography{matching}
\end{document}